\documentclass[conference, 10pt, a4paper, final]{IEEEtran}

\usepackage{cite}

\usepackage{graphicx}

\usepackage{geometry}
\usepackage[cmex10]{amsmath}
\usepackage{multirow}
\usepackage[]{algorithm2e}
\usepackage{array}
\usepackage{bbold}
\usepackage{stfloats}
\usepackage{mathtools}
\DeclarePairedDelimiter{\ceil}{\lceil}{\rceil}
\usepackage{stfloats}
\usepackage{url}
\usepackage{amssymb}
\usepackage{amsthm}
\usepackage{verbatim}
\usepackage{enumerate}
\usepackage{nicefrac}
\usepackage{subcaption}

\theoremstyle{definition}
\newtheorem{definition}{Definition}
\theoremstyle{remark}

\theoremstyle{definition}
\newtheorem{theorem}{Theorem}
\newtheorem{lemma}{Lemma}

\DeclareMathOperator*{\argmin}{argmin}
\IEEEoverridecommandlockouts
\begin{document}

\title{Cournot-Nash Equilibria for Bandwidth Allocation under Base-Station Cooperation}

\author{J.S. Gomez,
        A. Vergne,
        P. Martins,
        L. Decreusefond,
        and~Wei~Chen%
\thanks{J.S. Gomez is with the Institut Mines-T\'el\'ecom, T\'el\'ecom ParisTech, CNRS, LTCI (Paris, France) and the Departement of Electronic Engineering of Tsinghua University (Beijing, China), e-mail: jean-sebastien.gomez@telecom-paristech.fr}%
\thanks{A. Vergne, P. Martins and L. Decreusefond are with the Institut Mines-T\'el\'ecom, T\'el\'ecom ParisTech, CNRS, LTCI (Paris, France), e-mails: \{anais.vergne, philippe.martins, laurent.decreusefond\}@telecom-paristech.fr} 
\thanks{Wei Chen is with the Department of Electronic Engineering of Tsinghua University (Beijing, China), e-mail: wchen@tsinghua.edu.cn}%
}


\markboth{}%
{Shell \MakeLowercase{\textit{et al.}}: Bare Demo of IEEEtran.cls for Journals}

\maketitle

\begin{abstract}
In this paper, a novel resource allocation scheme based on discrete Cournot-Nash equilibria and optimal transport theory is proposed. The originality of this framework lies in the joint optimization of downlink bandwidth allocation and cooperation between base stations. A tractable formalization is given in the form of a quadratic optimization problem. A low complexity approximate solution is derived and theoretically characterized. Simulations highlight the existence of an optimal working point, that maximizes user satisfaction ratio and network load. The impact of the network deployment on the optimum is numerically investigated, thanks to the $\beta$-Ginibre model. Indeed, base stations are assumed to be drawn according to $\beta$-Ginibre point processes. Numerical analysis shows that the network performance increases with $\beta$ going to one.
\end{abstract}


 \ifCLASSOPTIONpeerreview
 \begin{center} \bfseries EDICS Category: 3-BBND \end{center}
 \fi

\IEEEpeerreviewmaketitle

\section{Introduction}


\IEEEPARstart{W}{ireless} networks have to tackle a major challenge: offering increasing user throughput while cost-efficiently allocating resources. Consequently, dynamic resource allocation adaptation to user traffic has been introduced in cellular networks. Strategies based on Markov processes \cite{Combes2012}, queuing theory \cite{ElSherif2014}, graph theory or game theory \cite{ZhuHan2005}  are used to finely tune bandwidth and power allocation. Nash bargaining theory has been used in this matter \cite{anchora2010}, assimilating the optimal resource allocation as a Nash equilibrium.

Another type of equilibria, the Cournot-Nash equilibria, has been defined by Antoine Augustin Cournot in 1838. He studied the situation of a spring water company duopoly. Each firm competes on the amount of their production output and decides at the same time which volume to produce in order to maximize its profit. This problem has been reformulated by Mas-Colell \cite{mas1984theorem} in probabilistic terms. Blanchet et al. were able to characterize existence and uniqueness of such equilibria in \cite{blanchet2015optimal} by taking advantage of properties of probability spaces and optimal transport theory exposed in the book of Villani \cite{villani2008optimal}. 

This paper focuses on a novel approach that jointly optimizes bandwidth allocation and cooperation between base stations. In the downlink scenario we consider, base stations are deployed according to a Poisson point process or a $\beta$-Ginibre point process. The SINR between each user and each base station is the only known information. Under these rough assumptions, we are able to solve the user bandwidth allocation and an optimal distribution of resources among cooperative base stations. Optima correspond to Cournot-Nash equilibria. We also show the link between Cournot-Nash equilibria and optimal transport theory and give a tractable mathematical formulation of the problem. A low complexity approximate optimal solution is also provided and characterized. Simulations reveal that there is an optimal working point of the network, where the user satisfaction ratio and the network load are equal. We finally compare the impact of the spatial deployment of the base stations, assuming that they are localized according to a $\beta$-Ginibre point process.

Resource allocation has been widely explored in literature. Many algorithms based on optimization have been described in \cite{xiaojun2006}. One example is the $\alpha$-fair resource allocation \cite{altman2008generalized} that gives a unified framework for optimization solution. Going one step further, optimal transport theory has been introduced in \cite{silva2013optimum} and \cite{mozaffari2016optimal}. This theory is used to shape cell boundaries and efficiently allocate power. Authors in \cite{silva2013optimum} introduce a congestion term, in order to modify the optimized solution, using the Wardrop equilibrium. Unlike optimization problems, this framework provides many mathematical tools to characterize the optimal solution. However, pure optimal transport solutions suffer from the fact that user demand for resources has to be known in order to compute the solution. Authors also limit their analysis to power allocation. On the contrary, the Cournot-Nash framework does not need an \textit{a priori} knowledge of the user spatial distribution. It can therefore solve the fair allocation of the bandwidth even in the case of outage. 
The impact of the regularity of the deployment on SINR has been studied and \cite{gomez2015} shows that the $\beta$-Ginibre point process is an eligible candidate to model cellular networks. 
To our knowledge, this is the first paper that uses the Cournot-Nash framework to tackle the joint resource allocation and cooperation problem and that investigates the impact of network deployment with a $\beta$-Ginibre point process model. 

This paper is organized as follows: in Section II, the system model is introduced and Cournot-Nash equilibria theory is applied to solve the resource allocation problem. In Section III, simulation results for different kinds of base station deployment are compared and analyzed. We conclude in Section IV.

\section{System model and problem formulation}
We consider a cellular wireless network composed of omnidirectional identical base stations drawn in the plane according to a certain point process (Poisson or $\beta$-Ginibre point process) of intensity $\lambda_{n}$. Users are drawn according to a Poisson point process of intensity $\lambda_{m}$. The state of the network is observed and assessed at a given moment. The downlink spectrum allocation problem where bandwidth is decomposed in blocks -i.e. resource block in LTE- is investigated. Using Shannon's capacity law, each user computes the number of resource blocks it wishes $N_j$, to fulfill its desired capacity $C_j$, based on the best SINR:
\begin{equation*}
N_j = \ceil[\Bigg]{ \frac{C_j}{W_{RB}\log_2(1 + \max_i (S\!I\!N\!R_{i,j}))}} ,
\end{equation*}
where $W_{RB}$ is the bandwidth of one resource block and $\ceil{x}$ is the ceil value of $x$.
In our scheme, a user can receive resource blocks from several base stations. 
 
When the number of users is large enough, the network has to share the available resource blocks among users and among base stations. The bandwidth allocation problem is thus divided into two sub-cases: 
\begin{itemize}
\item Knowing $N_j$, how many resource blocks does the network allocate to the $j$th user?
\item Knowing the number of resource blocks allocated to the $j$th user, from which base stations should they be transferred? 
\end{itemize}
The second sub-problem can be addressed with to optimal transport theory.
\begin{table}[h]
\renewcommand{\arraystretch}{1.3}
\caption{Notations}
\label{table:1}
\centering
\begin{tabular}{|c|c|}
\hline
\multicolumn{1}{|c|}{\multirow{2}{*}{$n$}}  & Number of base stations \\
\multicolumn{1}{|c|}{}  & drawn according to the chosen point process \\ \hline
\multicolumn{1}{|c|}{\multirow{2}{*}{$m$}}  & Number of users \\ 
\multicolumn{1}{|c|}{}  & drawn according to the chosen point process \\ \hline
\multicolumn{1}{|c|}{\multirow{2}{*}{$N_t$}} & Total available number \\ 
\multicolumn{1}{|c|}{}  & of resource blocks in the network \\ \hline
\multicolumn{1}{|c|}{\multirow{2}{*}{$\mu_i$}}  & Proportion of the total available \\ \multicolumn{1}{|c|}{}  & resource blocks at the $i$th base station \\ \hline
\multicolumn{1}{|c|}{\multirow{2}{*}{$\nu_j$}}  & Proportion of the total allocated  \\ \multicolumn{1}{|c|}{}  &  resource blocks at the $j$th user \\ \hline
\multicolumn{1}{|c|}{\multirow{2}{*}{$N_j$}}  & Number of resource blocks\\ \multicolumn{1}{|c|}{}  & requested by the $j$th user \\ \hline
\multicolumn{1}{|c|}{\multirow{2}{*}{$\gamma_{ij}$}}  & Proportion of resources allocated  \\
\multicolumn{1}{|c|}{}  & from the $i$th base station to the $j$th user \\ \hline
\multicolumn{1}{|c|}{\multirow{2}{*}{ SINR$_{ij}$}} & Measured SINR between \\ 
\multicolumn{1}{|c|}{}  & the $i$th base station and the $j$th user \\ \hline
\end{tabular}
\end{table}

\subsection{Optimal transport and base stations cooperation}
In 1781, Monge first described the optimal transport problem. One has to transfer sand from a pile of sand to a hole in the ground. Knowing the shape of the pile and of the hole, what are the paths taken by each grain of sand that minimize the energy used to transfer the pile to the hole? 
Assimilating $\boldsymbol\mu=\left(\mu_1,  \ldots, \mu_{n} \right)$, $\boldsymbol\nu=\left(\nu_1,\ldots, \nu_{m} \right)$ and $\boldsymbol\gamma=\left(\gamma_{1,1}, \ldots \gamma_{n,m} \right)$ as discrete probability measures respectively $\mu \in \mathcal{P}(\left[1,n\right])$, $\nu \in \mathcal{P}(\left[1,m \right])$ and $\gamma \in \mathcal{P}\left( \left[1, m\right] \times \left[1,n\right] \right)$, the resource transfer problem can be described by the discrete transport problem, where the pile is identified by $\mu$, the hole is identified by $\nu$ and the quantity transferred is given by $\gamma$. The transport cost between the $i$th and $j$th entities is defined by $c_{i,j} = \mathrm{SINR}_{ij}^{-1}$.
$\gamma$ is naturally the joint probability density of marginals $\mu$ and $\nu$.

\begin{definition}
The optimal transfer policy is given by the linear optimization problem:
\begin{equation*}
\gamma^*  = \argmin_{\gamma \in \Pi(\mu, \nu)} \sum_{(i,j)}c_{i,j}\,\gamma_{ij},
\end{equation*} 
where $\Pi(\mu, \nu)$ is the space of the joint probability measures of marginals $\mu$ and $\nu$.
\end{definition}
In other words, $\gamma^*$ solves the transportation problem between $\mu$ and $\nu$, and verifies: 
\begin{align*}
&\forall 1 \! \leq \! i \! \leq \! n, \; \sum^{m}_{j=1}\!\gamma^*_{ij} = \mu_i, \\
&\forall 1 \! \leq \! j \! \leq \! m,  \; \sum^{n}_{i=1}\!\gamma^*_{ij} = \nu_j, \\
&\forall (i,j), \; \gamma_{i,j} \geq 0.
\end{align*}
Collapsing $\boldsymbol\gamma$ and $\mathbf{c}$ into a vector form, the previous optimization problem can be rewritten in this form:
\begin{equation*}
\boldsymbol\gamma^* = \argmin_{\boldsymbol\gamma} {}^t \mathbf{c} \cdot \boldsymbol\gamma,
\end{equation*}
such that:
\begin{align*}
&T_{n} \boldsymbol\gamma = \boldsymbol\mu, \\
&T_{m} \boldsymbol\gamma = \boldsymbol\nu, \\
&\forall  1 \! \leq \! l  \! \leq \! nm, \; \gamma_l \geq 0,
\end{align*}
where: 
\begin{align*}
&T_{n} = \mathbf{1}_{1,m} \otimes \mathbf{Id}_{n}, \\
&T_{m} = \mathbf{Id}_{m} \otimes \mathbf{1}_{1,n}.
\end{align*}
$\otimes$ denotes the Kronecker product, $\mathbf{1}_{n,m}$ is the matrix of ones with $n$ lines and $m$ columns and $\mathbf{Id}_n$ is the square identity matrix on $\mathbb{R}^n$.

Since  $\gamma^*$ is a probability measure, this linear programming problem takes place on a compact set. Optimal solutions therefore exist. However, before applying optimal transport theory, we must first obtain the probability measure $\nu$ representing the user demand. This problem is solved with Cournot-Nash equilibria.

\subsection{Exact Cournot-Nash equilibria}
\begin{definition}
Using previous notations, the Cournot-Nash equilibria are the joint density probabilities $\gamma^*$ such that their second marginal $\nu^*$ verifies:
\begin{equation*}
\nu^* =  \argmin_{\nu \in \mathcal{P}\left([1,m] \right)} W_c(\mu,\nu)\!+\!s(\nu),
\end{equation*}
where:
\begin{equation*}
W_c(\mu,\nu)\!=\!\inf_{\gamma \in \Pi(\mu, \nu)} \!\mathbf{c} \cdot \boldsymbol\gamma,
\end{equation*}
\begin{equation*}
s(\nu) = {}^t \left(\boldsymbol\nu - \frac{\mathbf{N}}{N_t} \right) \cdot \left(\boldsymbol\nu - \frac{\mathbf{N}}{N_t}\right),
\end{equation*}
and $\mathbf{N} = \left(N_1, \ldots, N_{m} \right)$.
\end{definition}
This definition of the Cournot-Nash equilibria is the one introduced by Blanchet et al. in \cite{blanchet2015optimal}. The first term, $W_c$, solves the optimal transport problem between the probabilities $\mu$ and $\nu$. It is also known as the Wasserstein distance between the probability measures $\mu$ and $\nu$. The second term $s(\nu)$ is the fairness term, it only depends on the probability measure $\nu$.
This Cournot-Nash problem can be reformulated into a quadratic optimization problem.
\begin{definition}
 Cournot-Nash equilibria are solutions of the following quadratic optimization problem :
\begin{equation*}
\boldsymbol\gamma^\mathbf{*} = \argmin_{\boldsymbol\gamma} {}^t\boldsymbol\gamma H \boldsymbol\gamma + {}^t \mathbf{L}\boldsymbol\gamma,
\end{equation*}
such that:
\begin{align*}
&T_{n} \boldsymbol\gamma = \boldsymbol\mu, \\
&T_{m} \boldsymbol\gamma \leq \mathbf{N}/N_t,\\
&\forall  1 \! \leq \! l  \! \leq \! nm, \; \gamma_l \geq 0,
\end{align*}
where:
\begin{align*}
&H = {}^t T_{m} T_{m} = \mathbf{Id_{m}} \otimes \mathbf{1}_{n,n}, \\
&\mathbf{L} =  \mathbf{c} - 2   T_{m} \frac{\mathbf{N}}{N_t}.
\end{align*}
\end{definition}
Since $H$ is a positive semi-definite matrix and $\gamma^*$ is a probability measure, the boundedness of the optimization domain is ensured. The existence of a solution is hence guaranteed. Such formulation is easily implementable in a common solver and allows numerical simulations on networks composed of up to half a thousand users in a reasonable amount of time as it will be shown in Section \ref{Sec_sim}.
\subsection{Approximate Cournot-Nash equilibria}
Thanks to the separability of the Cournot-Nash objective function, one can interpret this Cournot-Nash equilibria as a superposition of the user allocation problem and the resource-transfer problem. This superposition structure is highlighted by the algebraic structure of $H$, due to the Kronecker product. Indeed, the $\mathbf{Id_m}$ factor gives the allocation of resources and the $\mathbf{1}_{n,n}$ factor gives for each user the optimal resource transfer. The quadratic formalization however cannot be fully separated due to the cost term $\mathbf{c}$.
Considering that most of the resources are allocated to the link of the minimal cost $c_{min,j}$ defined by $c_{min,j}\!=\!\min_{ i} c_{i,j}  $, the classical approach of the problem consists in:
\begin{itemize}
\item first, solving the resource allocation problem at a user level,
\item second, routing the allocated resources among cooperating base stations to attain the final user. 
\end{itemize}
\subsubsection{Resource allocation algorithm}
Therefore, the simplified quadratic problem function is derived:
\begin{equation*}
\boldsymbol\nu^* = \argmin_{\boldsymbol\nu} {}^t\boldsymbol\nu H \boldsymbol\nu + {}^t \mathbf{L}\boldsymbol\nu,
\end{equation*}
such that:
\begin{equation*}
{}^t\mathbf{1}_{1,m} \cdot \boldsymbol\nu = 1 \; \mathrm{and} \; \nu_j \geq 0, 
\end{equation*}
with:
\begin{equation*}
H = \mathbf{Id}_{m}, \; \mathrm{and} \;
\mathbf{L} =  \mathbf{c}_{min} - 2\frac{\mathbf{N}}{N_t}.
\end{equation*}
\begin{theorem}
\label{th1}
The solution $\boldsymbol\nu^{*}$ of the above simplified optimization problem is unique and is of the form:
\begin{equation*}
\boldsymbol\nu^* = \boldsymbol\nu^0 - \frac{{}^t\mathbf{u}(\boldsymbol\nu^0 - \mathbf{M})}{(m\!-\!k)}  \mathbf{u},
\end{equation*}
where $k$ is the number of zero coordinates of $\boldsymbol\nu^{*}$, $\mathbf{u} = \mathbf{1}_{m-k,1}$ $\boldsymbol\nu^0 = - \mathbf{L}/2$ and $\mathbf{M} = \mathbf{u} /(m\!-\!k)$.
\end{theorem}
\begin{proof}
The theorem is proven in Appendix \ref{Appendice1}. 
\end{proof}

\subsubsection{Cooperation}
Solving the resource transfer problem is equivalent to solve the optimal transport problem:
\begin{equation*}
\boldsymbol\gamma^\mathbf{*} = \argmin_{\boldsymbol\gamma} {}^t \mathbf{c} \cdot \boldsymbol\gamma,
\end{equation*}
such that:
\begin{align*}
T_{n} \boldsymbol\gamma &= \boldsymbol\mu, \\
T_{m} \boldsymbol\gamma &= \boldsymbol\nu^*.
\end{align*}

\subsubsection{Complexity}
\begin{theorem}
An approximate optimum can be found in polynomial time.
\end{theorem}
\begin{proof}
In order to compute approximate Cournot-Nash equilibria, one must first solve the allocation problem and then the optimal transport problem. Since the allocation problem involves at most $m$ projections, its complexity is in $\mathcal{O}(m^2)$. The optimal transport part is a linear programming optimization problem and is known to be solved in polynomial time.  
\end{proof}

\subsubsection{Algorithm}
Algorithm \ref{alg1} is derived from the proof of Theorem \ref{th1}. The optimal allocation $\nu^{\mathbf{*}}$ is first computed, then the optimal transport $\gamma^*$ between the two discrete probability measures $\mu$ and $\nu^*$ is derived by linear programming. The algorithm is centralized and iterative. The while-loop converges as the dimensions of the projective space is strictly decreasing and bounded by one. 

\begin{algorithm}
 \KwData{$\mathbf{c}, \mathbf{c_{min}}, \mathbf{N}, N_t, T_m, T_n$.}
 \KwResult{$\boldsymbol\gamma^{\mathbf{*}}$.}
 Initalize (k, $\boldsymbol\nu^0$, $\mathbf{M}$, $\mathbf{u}$, $\boldsymbol\nu^*$)\;
 \While{ $\exists \nu^{*}_j < 0$ }{
	Project ($\boldsymbol\nu^0$, $\mathbf{M}$, $\mathbf{u}$, $\boldsymbol\nu^*$) on the space of strictly positive coordinates of $\boldsymbol\nu^*$\;
	Let $k$ be the number of negative coordinates of $\boldsymbol\nu^*$\;	
	Project $\boldsymbol\nu^*$ on the new hyperplane defined by (k, $\boldsymbol\nu^0$, $\mathbf{M}$, $\mathbf{u}$, $\boldsymbol\nu^*$)\;
 }
 $\boldsymbol\gamma^{\mathbf{*}}= $ LinearProg$(\boldsymbol\mu, \boldsymbol\nu^{*}, \mathbf{c}, T_n, T_m)$\;
  \caption{Approximate solution algorithm}
\label{alg1}
\end{algorithm}
\subsection{Cournot-Nash equilibria and system optimum}
On a system level, three indicators are analyzed in function of the number of users in the network:
\begin{itemize}
\item The user satisfaction ratio: $$r_u = \frac{N_t}{m} \sum_{j=1}^m \frac{\nu_j}{N_j}.$$ It is the mean ratio between the number of resource blocks allocated by the network to each user and the number of resources requested by each user.
\item The network load: $$r_n = \sum_{j=1}^m \nu_j.$$ It is the proportion of total available resources used in the network. 
\item The cooperation proportion: $$r_c = \frac{1}{m} \sum_{j=1}^m \mathbb{1}\!\left( \left(\sum_{i=1}^n  \mathbb{1}(\gamma_{i,j} \neq 0) \right) \neq 1 \right),$$ where $\mathbb{1}$ is the indicator function. It is the proportion of users that receive resource blocks from multiple base stations. 
\end{itemize}
We define the optimum network working point as the intersection of user satisfaction curve and the network load curve. In the next section, this point is identified for both the exact and the approximate Cournot-Nash solutions. Its relative position is investigated under several network deployment schemes.
\section{Numerical analysis}
\label{Sec_sim}

\subsection{Simulation parameters}

\begin{table}[h]
\renewcommand{\arraystretch}{1.3}
\caption{Simulation parameters}
\label{table:2}
\centering
\begin{tabular}{|c|c|}
\hline
$\lambda_n$ & 10 per unit square\\ \hline
$\lambda_m$ & from 10 to 500 per unit square\\ \hline
$RB_{max}$  & 10\\ \hline
$C$ & 500 kB/s \\ \hline
$W_{RB}$ & 180 kHz\\ \hline
Resource blocks per base station & 100 \\ \hline
Path-loss exponent & 3 \\ \hline
Shadowing & 10 dB \\ \hline
\end{tabular}
\end{table}

Simulation parameters are summarized in Table \ref{table:2}. We assume that each base station reuses all the resource blocks. All antennas are omnidirectional and emit at the same power level. Base station locations are drawn according to a Poisson point process of intensity $\lambda_n$. Users locations are drawn according to a Poisson point process with intensity $\lambda_{m}$. Each user asks for the same capacity $C$.
The number of resource blocks per user is limited to $RB_{max}$. Therefore the number of resource blocks requested per user is given by:
\begin{equation*}
N_j\!=\!\max\left(\ceil[\Bigg]{\frac{C}{ W_{RB}\log_2(1 + \max_j (S\!I\!N\!R_{i,j})}}\!,\!RB_{max} \right).
\end{equation*}

\subsection{Exact vs. approximate Cournot-Nash solution}

\begin{figure}[h]
\centering
\includegraphics[width=88mm]{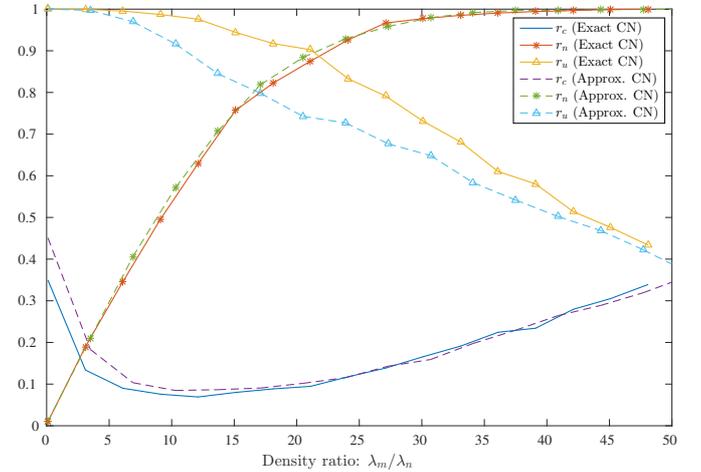}
\caption{Exact Cournot-Nash vs. approximate Cournot-Nash.}
\label{Figure:1}
\end{figure}

In Figure \ref{Figure:1}, $r_n$, $r_u$ and $r_c$ are plotted for the exact Cournot-Nash equilibria in function of the density ratio $\lambda_m / \lambda_n$ in solid lines. Using Matlab \texttt{quadprog} function, an iteration for the maximum number of users takes about 3 seconds to compute on a late 2014, 8 cores CPU laptop computer. This figure was produced with 500 iterations. The optimum working point of the network is reached for a density ratio of 21 and for a user satisfaction ratio (or a network load) of 89\%. The cooperation proportion reaches a minimum in the neighborhood of the optimum working point, with about 10\% of the users under base station cooperation.  

A comparison between the approximate solutions and the exact solutions is also given in Figure \ref{Figure:1}. The optimal transport was solved with the \texttt{intlinprog} function. One iteration for a density of 500 users and 10 base stations per unit square is computed in about 350 ms. The network optimum working point is reached for a density ratio of 17 and for a user satisfaction ratio (or a network load) of 80\%. The approximate algorithm thus proves to be a pessimistic bound of the exact Cournot-Nash solution, that can be used for an under-estimate of the network performance. It is a good trade-off between computational complexity and precision, since computation is about ten times faster than the exact algorithm whereas the error made is only of 10\% on the indicator. The cooperation proportion and the network load behaviors are similar to the exact curves.

\subsection{Impact of network deployment on the optimum network working point}

We consider networks composed of antennas drawn according to a $\beta$-Ginibre or Poisson point process with the same intensity $\lambda_{n}$. The $\beta$-Ginibre point process is a repulsive point process, which regularity can be set with the parameter $\beta$. A $\beta$-Ginibre point process is obtained after a thinning of a Ginibre point process. Each point of the Ginibre point process is independently selected with a probability $\beta$. If $\beta$ goes to $0$, the point process tends to a Poisson point process (corresponds to a uniform network deployment). If $\beta = 1$, then the point process corresponds to a Ginibre point process (corresponds to a regular network deployment). A way to simulate a Ginibre point process (and therefore a $\beta$-Ginibre point process) is given in \cite{decreusefond2013ginibre}.

\begin{figure*}%
\centering
\begin{subfigure}{1\columnwidth}
\includegraphics[width=\columnwidth]{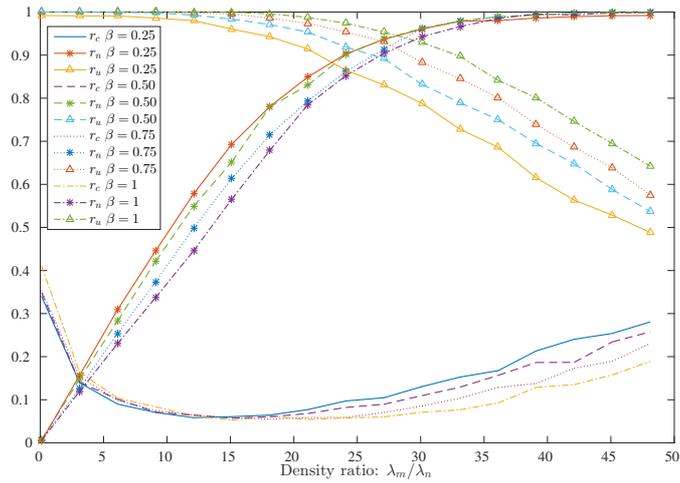}%
\caption{Exact Cournot-Nash.}%
\label{figure:2}%
\end{subfigure}\hfill%
\begin{subfigure}{1\columnwidth}
\includegraphics[width=\columnwidth]{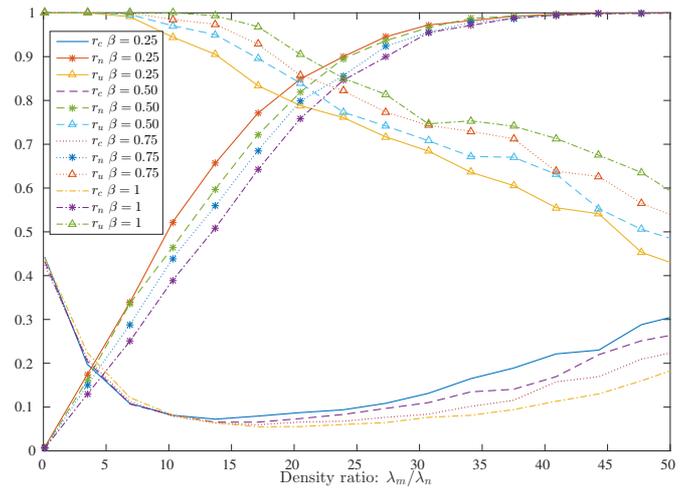}%
\caption{Approximate Cournot-Nash.}%
\label{figure:3}%
\end{subfigure}\hfill%
\caption{Equilibria obtained for $\beta=$0.25, 0.50, 0.75 and 1.}
\label{fig:23}
\end{figure*}

In Figures \ref{figure:2} and \ref{figure:3}, the impact of regularity is studied. Curves for Poisson  and $\beta$-Ginibre point processes are plotted for the exact and the approximate Cournot-Nash equilibria. Four $\beta$-Ginibre point processes are considered with four values of $\beta$:  0.25, 0.50, 0.75 and 1. Results are given in Table \ref{table:3}. 

\begin{table}[h]
\centering
\caption{Optimum network working points in function of $\beta$.}
\label{table:3}
\begin{tabular}{c|cc|cc|}
\cline{2-5}
                                     & \multicolumn{2}{c|}{Exact CN}                      & \multicolumn{2}{c|}{Approx. CN}                    \\ \hline
\multicolumn{1}{|c|}{Point Process}  & $\nicefrac{\lambda_m}{\lambda_n}$ & $r_u$ or $r_n$ & $\nicefrac{\lambda_m}{\lambda_n}$ & $r_u$ or $r_n$ \\ \hline
\multicolumn{1}{|c|}{Poisson}        & 21                                & 88\%           & 17                                & 80\%           \\ \hline
\multicolumn{1}{|c|}{$\beta = 0.25$} & 22.5                              & 88\%           & 19                                & 82\%           \\ \hline
\multicolumn{1}{|c|}{$\beta = 0.50$} & 25                                & 90\%           & 21                                & 84\%           \\ \hline
\multicolumn{1}{|c|}{$\beta = 0.75$} & 27.5                              & 92\%           & 22.5                              & 85\%           \\ \hline
\multicolumn{1}{|c|}{$\beta = 1$}    & 29                                & 94\%           & 24                                & 87\%           \\ \hline
\end{tabular}
\end{table}

For both exact and approximate Cournot-Nash equilibria, the density ratio and the user satisfaction of the optimum working point jointly increase with the value of $\beta$. This can be explained as the overall SINR quality in the network increases with the regularity of the deployment \cite{Nakata:2014:SSM:2643750.2643953}.

\section{Conclusion}
A novel resource allocation scheme under cooperation based on Cournot-Nash equilibria has been introduced. An exact as well as an approximate fast computable solution have been provided. Numerical analysis has shown the existence of an optimum network working point, where network load and user satisfaction ratio are jointly maximized. The cooperation proportion, is minimum in the neighborhood of the optimum working point. Impact of the network deployment has been investigated. The more regular the network is, the better the performance is. 

\appendices

\section{Proof of theorem 1}
\label{Appendice1}
\begin{lemma}
The simplified optimization problem can be transformed into a hypersphere equation.
\label{lemma_1}
\end{lemma}
\begin{proof}
The simplified optimization problem can be written in the following form:
\begin{equation*}
\boldsymbol\nu^{\mathbf{*}} = \argmin_{\boldsymbol\nu} {}^t\boldsymbol\nu H \boldsymbol\nu + {}^t \mathbf{L}\boldsymbol\nu + \frac{1}{4} {}^t \mathbf{L} \mathbf{L},
\end{equation*}
such that:
\begin{equation*}
{}^t\mathbf{1}_{1,m} \cdot \boldsymbol\nu = 1 \; \mathrm{and} \; \nu_j \geq 0, 
\end{equation*}
We denote by $\mathcal{C}$, the convex hull defined by the constraints of this optimization problem.
The added constant does not modify the optima and therefore this problem is equivalent to the simplified optimization problem. Furthermore, the utility function is the equation of an hypersphere of center $\boldsymbol\nu^0 = - \mathbf{L} / 2$ and the objective value is its radius.
\end{proof}
Thanks to Lemma \ref{lemma_1}, the optimum $\boldsymbol\nu^{\mathbf{*}}$ is given by the intersection of the minimal radius hypersphere of center $\boldsymbol\nu^0 = - \mathbf{L} / 2$ and of $\mathcal{C}$. 
Let $\mathcal{H}$ be the hyperplane defined by:
\begin{equation*}
\mathcal{H} = \left\lbrace \mathbf{x} \in \mathbb{R}^{m} \mid {}^t \mathbf{1}_{m,1}  \mathbf{x} = 1 \right\rbrace.
\end{equation*}
$\mathcal{C}$ is included in the hyperplane $\mathcal{H}$. 
Let $\boldsymbol\nu^*$ be the orthogonal projection of $\boldsymbol\nu^0$ on $\mathcal{H}$. Two cases can be distinguished:
\begin{enumerate}
\item $\boldsymbol\nu^{*}$ has no strictly negative coordinates.
\item $\boldsymbol\nu^{*}$ has some strictly negative coordinates.
\end{enumerate}
In the first case, $\boldsymbol\nu^{*}$ is the tangent point between $\mathcal{C}$ and the hypersphere. Since $\boldsymbol\nu^{*}$ is the orthogonal projection of $\boldsymbol\nu^0$ on $\mathcal{C}$, it also minimizes the radius of the hypersphere that intersect $\mathcal{C}$. The optimum is given by:
\begin{equation*}
\boldsymbol\nu^{*} = \boldsymbol\nu^0 - \frac{ {}^t \mathbf{u} (\boldsymbol\nu^0  - \mathbf{M}) }{m}  \mathbf{u},
\end{equation*}
where $\mathbf{M} = \mathbf{1}_{m,1} / m$ and $\mathbf{u} = \mathbf{1}_{m,1}$. If all coordinates are positive, then the optimum has been reached.\\ 
In the second case (indexing from $1$ to $m-k$ the strictly positive coordinates, where $k$ is the number of negative coordinates), the positivity constraints $m-k+1$ to $m$ are saturated. $\boldsymbol\nu^*$ is in $\mathcal{H}$ but outside $\mathcal{C}$. Therefore, $\nu_{m-k+1}^* \ldots \nu_{m}^*$ are set to zero and $\nu_{1}^* \ldots \nu_{m-k}^*$ have to be computed. $\mathbf{M}, \mathbf{u}$ and $\boldsymbol\nu^0$ are first projected on the non-null subspace:
\begin{align*}
\forall& 1  \leq j \leq m\!-\!k, \;  \mathrm{M}_{j} = 1/(m\!-\!k),  \\
\forall& m-k+1  \leq j \leq m, \; \nu_{j}^0 = 0, \mathrm{M}_j = 0, \mathrm{u}_j = 0.
\end{align*}
Then the optimum is calculated in the non-null subspace:
\begin{equation*}
\boldsymbol\nu^* = \boldsymbol\nu^0 - \frac{{}^t\mathbf{u}(\boldsymbol\nu^0 - \mathbf{M})}{(m\!-\!k)}  \mathbf{u}.
\end{equation*}
In this case, the previous operations must be repeated until all coordinates are positive.

Uniqueness of the solution is ensured by the fact that the optimal solution is the orthogonal projection  of the center of an hypersphere. 

\bibliographystyle{IEEEtran}
\bibliography{IEEEabrv,biblio}
\end{document}